\documentclass[12pt,fleqn]{article}

\usepackage{amsmath}
\usepackage{amsthm}
\usepackage{amssymb}
\usepackage{amsfonts}
\usepackage{epsf}
\usepackage{graphicx}
\usepackage{hyperref}
\usepackage{color}
\usepackage[lofdepth,lotdepth]{subfig}
\usepackage{bm}
\usepackage{booktabs}

\newtheorem{lemma}{Lemma}

\newtheorem{theorem}{Theorem}

\newtheorem{prop}{Proposition}

\theoremstyle{remark}

\newtheorem{remark}{Remark}
\theoremstyle{definition}

 \numberwithin{equation}{section}
\DeclareMathOperator{\sgn}{sgn}

\newcommand{\D}{\displaystyle}

\numberwithin{equation}{section}

\newcounter{comment}
\begin{document}

\title{The distribution function for the maximal height of $N$ non-intersecting Bessel paths}
\date{\today}
\author{Dan Dai$^{\ast}$ and Luming Yao$^{\dag}$}

\maketitle
\begin{abstract}
    In this paper, we consider $N$ non-intersecting Bessel paths starting at $x=a\geq 0$, and conditioned to end at the origin $x=0$. We derive the explicit formula of the distribution function for the maximum height. Depending on the starting point $a>0$ or $a=0$, the distribution functions are also given in terms of the Hankel determinants associated with the multiple  discrete orthogonal polynomials or discrete orthogonal polynomials, respectively.
\end{abstract}

\vspace{5cm}


%%%%%%%%%%%%%%%%%%%%%%%%%%%%%%%%%%%%%%%%%%%%%%%%%%%%%%%%%%%%%%%%%%

\vspace{5mm}

\hrule width 65mm

\vspace{2mm}

\begin{description}

\item \hspace*{5mm}$\ast$ Department of Mathematics, City University of
Hong Kong, Hong Kong. \\
Email: \texttt{dandai@cityu.edu.hk}

\item \hspace*{5mm}$\dag$ Department of Mathematics, City University of
Hong Kong, Hong Kong. \\
Email: \texttt{lumingyao2-c@my.cityu.edu.hk} %(corresponding author)

\end{description}

\newpage

\section{Introduction} \label{sec:intro}

Random walk models play an important role in various areas in physics, chemistry and computer sciences. For example, they are introduced to explain basic concepts in statistical physics \cite{reif} and stochastic algorithms \cite{motwani}. When a proper diffusion scaling limit is taken, the random walk models are reduced to  Brownian motion models. The standard one-dimensional Brownian motion $\{W_t : t \ge 0\}$ initiated at $x$ is a stochastic process with the following properties (cf. \cite{borodin}):
\begin{itemize}
    \item {$W_0=x$ almost surely;}
    \item{$W_t$ is continuous almost surely;}
    \item {for all $0=t_0<t_1<\cdots<t_n$, the increments $W_{t_i}-W_{t_{i-1}}$ are independent and normally distributed, with $\mathbf{E}(W_{t_i}-W_{t_{i-1}})=0$ and $\mathbf{E}(W_{t_i}-W_{t_{i-1}})^2=t_i-t_{i-1}$.}
\end{itemize}
There is a vast literature in the study of Brownian motions, among which people are interested in the maximal height of the outermost path in different non-intersecting Brownian motions. For example, let $0 \leq b_1^{(BE)} (t) < b_2^{(BE)} (t) < \cdots < b_N^{(BE)}(t)$ be $N$ non-intersecting Brownian excursions, i.e. non-intersecting Brownian motions with an absorbing wall located at $x=0$, conditioned to
have the same starting point $x = 0$ and return to the origin at the end. The transition probability for a single Brownian motion with an absorbing wall at $x=0$, passing from $x$ to $y$ over the time interval $t$, is given by
\begin{equation} \label{absorb-prob}
    p_{abs}(t, y|x)=\frac{1}{\sqrt{2\pi t}}\left(e^{-\frac{(x-y)^2}{2t}} - e^{-\frac{(x+y)^2}{2t}}\right).
\end{equation}
Let $M$ be a positive constant. Then, the probability that the maximal height of the outermost path $\max\limits_{0<t<1}b_N^{(BE)}(t)$ is less than $M$ is given explicitly as follows:
\begin{equation} \label{BE-case}
\begin{split}
  &\mathbb{P}(\max_{0<t<1}b_N^{(BE)}(t)<M) \\
  & \qquad =\frac{2^{-N/2} \pi^{2N^2+N/2}}{M^{N(2N+1)}N! \prod_{k=0}^{N-1}(2k+1)!}\sum_{\mathbf x \in \mathbb{Z}^N}(\Delta(\mathbf {x^2}))^2 \left(\prod_{j=1}^N x_j^2 \right) \exp \biggl\{\frac{- \pi^2}{2M^2}\sum_{j=1}^N x_j^2 \biggr\},
\end{split}
\end{equation}
where $\Delta(\mathbf x)$ is the Vandermonde determinant
\begin{equation} \label{vandermond}
\Delta(\mathbf x)=\prod_{1 \le j <k \le N}(x_k-x_j) \qquad \textrm{with} \quad \mathbf x=(x_1, x_2, \cdots, x_N).
\end{equation}
The formula \eqref{BE-case} was first obtained by Schehr et al. \cite{gregory2008} through  a path integral approach. Later, it was derived by using the Karlin-McGregor formula in Kobayashi et al. \cite{naoki2008}, and by lattice paths in  Feierl \cite{Feierl2012}. One may also consider $N$ non-intersecting Brownian motions $0 \leq b_1^{(R)} (t) < b_2^{(R)} (t) < \cdots < b_N^{(R)}(t)$ with a  reflecting wall located at $x=0$, conditioned to
have the same starting point $x = 0$ and return to the origin at the end. In this case, the transition probability is
\begin{equation} \label{reflecting-prob}
    p_{ref}(t, y|x)=\frac{1}{\sqrt{2\pi t}}\left(e^{-\frac{(x+y)^2}{2t}}+e^{-\frac{(x-y)^2}{2t}}\right),
\end{equation}
and the probability for the maximal height is given by
\begin{equation}\label{r}
\begin{split}
   & \mathbb{P}(\max_{0<t<1}b_N^{(R)}(t)<M) \\
   & \qquad =\frac{2^{-N/2} \pi^{2N^2-3N/2}}{M^{N(2N-1)}N! \prod_{k=0}^{N-1}(2k)!}\sum_{\mathbf x \in \{\mathbb{Z}-1/2\}^N}(\Delta(\mathbf {x^2}))^2\exp \biggl\{\frac{- \pi^2}{2M^2}\sum_{j=1}^N x_j^2 \biggr\};
\end{split}
\end{equation}
see Liechty \cite[Eq. (1.18)]{Karl2012}. It is worth mentioning that the above two probabilities \eqref{BE-case} and \eqref{r} can be put in terms of Hankel determinants of discrete Gaussian orthogonal polynomial. Based on this observation, Liechty \cite{Karl2012} applied the Riemann-Hilbert approach to rigorously prove that,  in a proper scaling limit as $N \to \infty$, the limiting distribution of the maximal height in both cases converges to the Tracy-Widom distribution for the largest eigenvalue of the Gaussian orthogonal ensemble. This justifies the formal derivations first given by  Forrester, Majumdar and Schehr \cite{For:Maj:Sch2011}.

In this paper, we focus on the distribution of the maximal height of the outermost path in $N$ non-intersecting Bessel paths. Recall that if $\{\mathbf{X}(t): t \ge 0\}$ is a $d$-dimensional Brownian motion, then the diffusion process
\begin{equation} \label{Rt-def}
R(t)=||\mathbf X(t)||_2=\sqrt{X_1(t)^2+\cdots+X_d(t)^2}, \qquad t \geq 0,
\end{equation}
is called a $d$-dimensional Bessel process or a Bessel process of order $\alpha$ with parameter $\alpha=\frac{d}{2}-1$; see \cite{borodin, karatzas, revuz}.  Bessel processes have important application in financial mathematics due to their close relation to financial models such as geometric Brownian motion and Cox-Ingersoll-Ross processes; see G\"{o}ing-Jaeschke and Yor \cite{going} and references therein. When $d=1$, or equivalently $\alpha=-\frac{1}{2}$, the Bessel process $R(t)$ reduces to the Brownian motion with a reflecting wall located at the origin \eqref{reflecting-prob}. While when $d=3$, or equivalently $\alpha=\frac{1}{2}$, the transition probability of $R(t)$ becomes
\begin{equation} \label{prob-alpha-1/2}
 p(t, y|x)=\frac{y}{x}\frac{1}{\sqrt{2\pi t}}\left(e^{-\frac{(x-y)^2}{2t}}-e^{-\frac{(x+y)^2}{2t}}\right),
\end{equation}
which is closely related to the Brownian motion with an absorbing wall located at the origin \eqref{absorb-prob}; see  Katori and  Tanemura \cite{makoto2004, makoto2007}. In recent years, there has been considerable interest in the study of non-intersecting (squared) Bessel paths;
for example, see \cite{Del:Kui:Zha2012, Katori2008, Kuijlaars2009, Kuijlaars2011}.

Similar to the non-intersecting Brownian motions mentioned above, we consider a model of $N$ non-intersecting Bessel paths: $\{b_j(t)\}_{j=1}^N$. All paths start at $x=a \geq 0$, remain non-negative and are conditioned to end at time $t=1$ at $x=0$. That is,
\begin{equation} \label{bm-model}
\begin{split}
 & b_1(0)=b_2(0)=\cdots=b_N(0)=a,\\
    & b_1(1)=b_2(1)=\cdots=b_N(1)=0,\\
    & 0 \le b_1(t)<b_2(t)< \cdots <b_N(t) \quad \text{for} \quad  0<t<1.
    \end{split}
\end{equation}
The transition probability for a single particle passes form $x$ to $y$ over the time interval $t$ is given by (cf. \cite{borodin, ito1974})
\begin{align}
    p(t,y|x)&=\frac{1}{t}\frac{y^{\alpha +1}}{x^{\alpha}}e^{-\frac{x^2+y^2}{2t}}I_{\alpha} \left( \frac{xy}{t} \right), \qquad x>0, \label{p1} \\
    p(t,y|0)&=\frac{y^{2\alpha +1}}{2^{\alpha}\Gamma{(\alpha+1)}t^{\alpha+1}}e^{-\frac{y^2}{2t}} \label{p11}
\end{align}
for $y \ge 0$, $t>0$ and $\alpha >-1$, where $I_{\alpha}$ is the modified Bessel function of the first kind
\begin{equation} \label{defI}
    I_{\alpha}(z)=\sum^{\infty}_{k=0}\frac{(\frac{z}{2})^{2k+\alpha}}{k! \, \Gamma (k+\alpha +1)}.
\end{equation}
One can see from \eqref{Rt-def} that when $d=2(\alpha +1)$ is an integer, the Bessel process is the distance to the origin of a $d$-dimensional standard Brownian motion.

The main results of the paper are stated in the next section.

\section{Statement of results}\label{sec:results}

\subsection{The maximal distribution}

Our first result is the explicit expression of the probability of the maximal height of the outermost path $\max\limits_{0<t<1}b_N(t)$.

\begin{theorem}\label{the:theorem}
For $\alpha > -1$, let
\begin{equation} \label{bessel-zeros}
x_{1,\alpha} < x_{2,\alpha} < \cdots < x_{n,\alpha} < \cdots
\end{equation}
be the zeros of the Bessel function $J_{\alpha}(x)$.     For $ N \ge 1$, $M \ge a$ and the starting point $a>0$, we have
    \begin{equation} \label{deter}
      \begin{split}
       &  \mathbb{P}(\max_{0<t<1}b_N(t)<M)  \\
       & \qquad =c_N(\alpha) M^{-\frac{N(3N+2\alpha+1)}{2}}  \det_{1 \le i, j \le N} \left[\sum_{n=1}^{\infty}\frac{(-1)^{i-1}x_{n, \alpha}^{i+2j+\alpha-3}J_{\alpha}^{(i-1)}(\frac{a}{M}x_{n,\alpha})e^{-\frac{x_{n, \alpha}^2}{2M^2}}}{J_{\alpha+1}^2(x_{n, \alpha})} \right],
    \end{split}
    \end{equation}
    where $c_N(\alpha)$ is a constant independent of $M$:
    \begin{equation} \label{cn-def}
      c_N(\alpha)=\frac{2^{\frac{N(3-N)}{2}} e^{\frac{a^2N}{2}}}{a^{\frac{N(N+2\alpha-1)}{2}}\prod_{j=1}^N \Gamma(j)}.
    \end{equation}
    When the starting point is the origin, i.e., $a = 0$, we have
    \begin{equation} \label{deter-a=0}
        \mathbb{P}(\max_{0 < t < 1} b_N(t) <M)= \widetilde{c}_N(\alpha) M^{-2N(N+\alpha)}\det_{1 \le i, j \le N} \left[\sum_{n=1}^{\infty}\frac{x_{n, \alpha}^{2i+2j+2\alpha-4}e^{-\frac{x_{n, \alpha}^2}{2M^2}}}{J_{\alpha+1}^2(x_{n, \alpha})} \right],
    \end{equation}
    where the constant $\widetilde{c}_N(\alpha)$ is given by
    \begin{equation} \label{cn2-def}
    \widetilde{c}_N(\alpha)=\frac{2^{2N- \alpha N-N^2}}{\prod_{j=1}^N \Gamma(j) \Gamma(\alpha +j)}.
    \end{equation}

\end{theorem}

\begin{remark}
Although the expressions \eqref{deter} and \eqref{deter-a=0} appear different, they indeed agree with each other. To see it, from properties of the Bessel function \cite[Eq. (10.6.2)]{dlmf}, one gets the following expression for the higher order derivative of $J_{\alpha}(z)$:
 \begin{equation}\label{kthde}
        (-1)^n z^n J_{\alpha}^{(n)}(z) = \sum_{k=0}^n c_{n,k} z^k J_{\alpha+k}(z) \qquad \textrm{for } n \in \mathbb{N},
    \end{equation}
    with the coefficient $c_{n,n}=1$. We replace $J_{\alpha}^{(i-1)}$ in \eqref{deter} by the right-hand side of the above formula, then apply row operations to eliminate $J_\nu$ for $\nu < \alpha+i-1$. This gives us
    \begin{equation}
        \begin{split}
            \mathbb{P}(\max_{0<t<1}b_N(t)<M) &=\frac{2^{\frac{N(3-N)}{2}} e^{\frac{a^2N}{2}}}{a^{\frac{N(N+2\alpha-1)}{2}}\prod_{j=1}^N \Gamma(j)} M^{-\frac{N(3N+2\alpha+1)}{2} }\\
            &\qquad \times \det_{1 \le i, j \le N} \left[\sum_{n=1}^{\infty}x_{n, \alpha}^{i+2j-3}\frac{J_{\alpha+i-1}(\frac{a}{M}x_{n,\alpha})x_{n, \alpha}^{\alpha}e^{-\frac{x_{n, \alpha}^2}{2M^2}}}{J_{\alpha+1}^2(x_{n, \alpha})} \right].
            \end{split}
            \end{equation}
            Since $J_\alpha(z) \sim \frac{z^\alpha}{2^\alpha \Gamma(\alpha +1)}$  as $z \to 0$,  the determinant in the above formula behaves like
            \begin{equation*}
            \det_{1 \le i, j \le N} \left[\sum_{n=1}^{\infty}x_{n, \alpha}^{i+2j-3}\frac{J_{\alpha+i-1}(\frac{a}{M}x_{n,\alpha})x_{n, \alpha}^{\alpha}e^{-\frac{x_{n, \alpha}^2}{2M^2}}}{J_{\alpha+1}^2(x_{n, \alpha})} \right]  \sim \det_{1 \le i, j \le N} \left[ \frac{(\frac{a}{2M})^{\alpha+i - 1}}{\Gamma(\alpha+j)}\sum_{n=1}^{\infty}x_{n, \alpha}^{2i+2j-4}\frac{x_{n, \alpha}^{2\alpha}e^{-\frac{x_{n, \alpha}^2}{2M^2}}}{J_{\alpha+1}^2(x_{n, \alpha})} \right]
            \end{equation*}
            as $a \to 0$. Then, \eqref{deter-a=0} follows from the above two formulae.
\end{remark}

\begin{remark}
    When the starting point $a$ is the origin and  the number of Bessel paths $N$ is equal to 1, the distribution formula \eqref{deter-a=0} reduces to
    \begin{equation}
        \mathbb{P}(\max_{0 < t < 1} b(t) <M)= \frac{2^{1- \alpha }}{\Gamma(\alpha +1)} M^{-2(\alpha+1)}\sum_{ n =1}^{\infty}\frac{x_{n, \alpha}^{2\alpha} e^{-\frac{x_{n, \alpha}^2}{2M^2}}}{J_{\alpha +1}^2(x_{n, \alpha})}.
    \end{equation}
    This agrees with the results obtained by Pitman and  Yor \cite[Eq. (5)]{jim1999}.
\end{remark}

\subsection{Relation to multiple orthogonal polynomials}

The maximum distributions in \eqref{BE-case} and \eqref{r} can be put in terms of Hankel determinants of discrete Gaussian orthogonal polynomial; see Liechty \cite{Karl2012}. For the present model, we have similar results. Depending on the starting point $a>0$ or $a=0$, the distribution functions are given in terms of the Hankel determinants associated with \emph{multiple discrete orthogonal polynomials} of type II or \emph{discrete orthogonal polynomials}, respectively. For more properties of multiple orthogonal polynomials, one may refer to \cite{Arv:Cou:Van2003} and \cite[Ch. 23]{Ismail-book}.

In the case where the starting point $a$ is positive, the corresponding two discrete weight functions for the multiple orthogonal polynomial are
\begin{align} \label{mop-weight}
    w_1(x)=\frac{x^{\frac{\alpha}{2}}J_{\alpha}(\frac{a}{M}\sqrt{x})}{J_{\alpha+1}^2(\sqrt{x})}e^{-\frac{x}{2M^2}},  \qquad
    w_2(x)&=\frac{x^{\frac{\alpha+1}{2}}J_{\alpha+1}(\frac{a}{M}\sqrt{x})}{J_{\alpha+1}^2(\sqrt{x})}e^{-\frac{x}{2M^2}}, \quad M > a,
\end{align}
which are supported on the nodes $\{x_{n, \alpha}^2 \}_{n=1}^\infty$ with $x_{n, \alpha}$ being zeros of the Bessel function $J_\alpha(x)$ given in \eqref{bessel-zeros}. The moments $m_k^{(j)}, j=1,2$, are given by
\begin{align} \label{mk-def}
    m_k^{(1)}=\sum_{x \in \{x_{n, \alpha}^2, n \in \mathbb{N} \}} x^k w_1(x), \qquad
    m_k^{(2)}=\sum_{x \in \{x_{n, \alpha}^2, n \in \mathbb{N} \}} x^k w_2(x).
\end{align}
The  Hankel determinant is
\begin{equation} \label{mop-hankel}
 H_N:= \det\begin{pmatrix}
    m_0^{(1)} & m_1^{(1)} & \dots  & m_{N-1}^{(1)} \\
    m_1^{(1)} & m_2^{(1)} & \dots  & m_{N}^{(1)} \\
    \vdots & \vdots &  \ddots & \vdots \\
    m_{n_1-1}^{(1)} & m_{n_1}^{(1)} &  \dots  & m_{N+n_1-2}^{(1)}\\
    m_0^{(2)} & m_1^{(2)} & \dots  & m_{N-1}^{(2)} \\
    m_1^{(2)} & m_2^{(2)} & \dots  & m_{N}^{(2)} \\
    \vdots & \vdots &  \ddots & \vdots \\
    m_{n_2-1}^{(2)} & m_{n_2}^{(2)} &  \dots  & m_{N+n_2-2}^{(2)}
\end{pmatrix}
\end{equation}
with $N=n_1+n_2$ and
\begin{equation}
    n_1=\begin{cases}
            \frac{N}{2}, & \mbox{if } N \mbox{ is even,}\\
            \frac{N+1}{2},& \mbox{if } N \mbox{ is odd,}
        \end{cases}
    \quad\text{and}\quad
    n_2=\begin{cases}
            \frac{N}{2}, & \mbox{if } N \mbox{ is even,}\\
            \frac{N-1}{2},& \mbox{if } N \mbox{ is odd.}
        \end{cases}
\end{equation}
In the case  where the starting point is located at the origin, the discrete weight function for the orthogonal polynomial is
\begin{equation}
    \widetilde{w}(x)=\frac{x^{\alpha}e^{-\frac{x}{2M^2}}}{J_{\alpha+1}^2(\sqrt{x})}, \qquad M>0,
\end{equation}
which is also supported on the nodes $\{x_{n, \alpha}^2 \}_{n=1}^\infty$. The corresponding Hankel determinant is
\begin{equation} \label{op-hankel}
    \widetilde{H}_N:= \det\begin{pmatrix}
        \widetilde{m}_0 & \widetilde{m}_1 & \dots  & \widetilde{m}_{N-1} \\
        \widetilde{m}_1 & \widetilde{m}_2 & \dots  & \widetilde{m}_{N} \\
        \vdots & \vdots &  \ddots & \vdots \\
        \widetilde{m}_{N-1} & \widetilde{m}_{N} &  \dots  & \widetilde{m}_{2N-2}
    \end{pmatrix},
\end{equation}
where the moments $\widetilde{m}_k$ are given by
\begin{equation} \label{mk-def-2}
    \widetilde{m}_k=\sum_{x \in \{x_{n, \alpha}^2, n \in \mathbb{N} \}} x^k \widetilde{w}(x).
\end{equation}

Then, we have the following result.
\begin{theorem} \label{thm-op}
    Let the Hankel determinants $H_N$ and $\widetilde{H}_N$ be defined in \eqref{mop-hankel} and  \eqref{op-hankel}, respectively. For $a>0$, we have
    \begin{equation}\label{hankel}
        \mathbb{P}(\max_{0<t<1}b_N(t)<M)=c_N(\alpha)(-1)^{\frac{N(N-3)+2n_1^2}{4}} M^{-\frac{N(3N+2\alpha+1)}{2}} H_N,
    \end{equation}
    where $c_N(\alpha)$ is given in \eqref{cn-def}.
    When $a=0$, we have
    \begin{equation}\label{hankel-a=0}
        \mathbb{P}(\max_{0<t<1}b_N(t)<M)=\widetilde{c}_N(\alpha) M^{-2N(N+\alpha)} \widetilde{H}_N,
    \end{equation}
    where $\widetilde{c}_N(\alpha)$ is given in \eqref{cn2-def}.
\end{theorem}

\subsection{Special cases: $ \alpha = \pm \frac{1}{2}$}

As mentioned in Section \ref{sec:intro},  the Bessel paths reduce to the Brownian motion with a reflecting wall at the origin when $\alpha=-\frac{1}{2}$. Indeed, recall that the Bessel functions reduce to the elementary functions when $\alpha = \pm \frac{1}{2}$:
\begin{equation}
    J_{-\frac{1}{2}}(x)=\sqrt{\frac{2}{\pi x}}\cos x,  \quad \quad J_{\frac{1}{2}}(x)=\sqrt{\frac{2}{\pi x}}\sin x \quad \textrm{and} \quad I_{-\frac{1}{2}}(x)=\sqrt{\frac{2}{\pi x}} \cosh x;
\end{equation}
see \cite[Sec. 10.16]{dlmf}. Then, the transition probability \eqref{p1} becomes  \eqref{reflecting-prob}. When the starting point $a =0$, the model \eqref{bm-model} exactly reduces to the non-intersecting Brownian motion with a reflecting wall at the origin, conditioned to
have the same starting point $x = 0$ and return to the origin at the end. Then, the probability for the maximum height given in \eqref{deter-a=0} also reduces to that of the Brownian motion with a reflecting wall given in \eqref{r}. To see it, note that the zeroes in \eqref{bessel-zeros} are $x_{n, \alpha}=(n-\frac{1}{2})\pi$, $n \in \mathbb{N}$ when  $\alpha = -\frac{1}{2}$. Then,
the determinant in \eqref{deter-a=0} becomes
\begin{equation}
    \det_{1 \le i, j \le N} \left[\sum_{n=1}^{\infty}x_{n, \alpha}^{2i+2j-4}\frac{x_{n, \alpha}^{2\alpha}e^{-\frac{x_{n, \alpha}^2}{2M^2}}}{J_{\alpha+1}^2(x_{n, \alpha})} \right]=\det_{1 \le i, j \le N} \left[\sum_{n=1}^{\infty}\frac{\pi^{2i+2j-3}}{2}(n-\frac{1}{2})^{2j+2i-4}e^{-\frac{(n-\frac{1}{2})^2 \pi ^2}{2M^2}} \right].
\end{equation}
Moreover, as $\Gamma(j) \Gamma(-\frac{1}{2} +j)=2^{2-2j}\sqrt{\pi} \, \Gamma(2j-1)$, the constant in \eqref{cn2-def} is
\begin{equation}
    \widetilde{c}_N(-\frac{1}{2})=\frac{2^{2N+1/2 N-N^2}}{\prod_{j=1}^N \Gamma(j) \Gamma(-\frac{1}{2} +j)}=\frac{2^{2N+1/2 N-N^2}}{\pi ^\frac{N}{2}2^{N-N^2}\prod_{j=0}^{N-1}(2j)!}.
\end{equation}
With the above two formulae, the probability in \eqref{deter-a=0} reduces to
\begin{align}
    \mathbb{P}(\max_{0 < t < 1} b_N(t) <M) & =\frac{2^{N/2} \pi^{2N^2-3N/2}}{M^{N(2N-1)} \prod_{j=0}^{N-1}(2j)!}\det_{1 \le i, j \le N} \left[\sum_{n=1}^{\infty}(n-\frac{1}{2})^{2i+2j-4}e^{-\frac{(n-\frac{1}{2})^2 \pi ^2}{2M^2}} \right] \nonumber \\
    & = \frac{2^{-N/2} \pi^{2N^2-3N/2}}{M^{N(2N-1)} \prod_{k=0}^{N-1}(2k)!}\det_{1 \le i, j \le N} \left[\sum_{x \in \{ \mathbb{Z}-1/2\}}x^{2i+2j-4}e^{-\frac{x^2 \pi ^2}{2M^2}} \right],
\end{align}
which agrees with \eqref{r}.

When $\alpha=\frac{1}{2}$, the Bessel paths are closely related to the Brownian motions with an absorbing wall at the origin; see the relation between their transition probabilities in \eqref{prob-alpha-1/2}. Although there is an additional factor $y/x$ in \eqref{prob-alpha-1/2}, the maximum distributions are exactly the same. Indeed, since the zeroes in \eqref{bessel-zeros} are $x_{n,\alpha}=n\pi$, $n \in \mathbb{Z}$ when  $\alpha = \frac{1}{2}$, we adopt similar computations as above and obtain from \eqref{deter-a=0}
\begin{equation}\label{alpha=1/2}
\mathbb{P}(\max_{0<t<1}b_N(t)<M) = \frac{2^{-N/2} \pi^{2N^2+N/2}}{M^{N(2N+1)} \prod_{k=0}^{N-1}(2k+1)!}\det_{1 \le i, j \le N} \left[\sum_{x \in \mathbb{Z} }x^{2i+2j-2}e^{-\frac{x^2 \pi ^2}{2M^2}} \right],
\end{equation}
which agrees with \eqref{BE-case}.

The rest of the paper is organized as follows. In Section \ref{Sec:preliminary}, we apply the Karlin-McGregor formula and express the probability $\mathbb{P}(\max\limits_{0 < t < 1} b_N(t) <M)$ in terms of a ratio of two determinants. Properties for some general determinants are also discussed. In Section \ref{sec:derivation}, we derive the asymptotics for the determinants appearing in  $\mathbb{P}(\max\limits_{0 < t < 1} b_N(t) <M)$. Then, our main results are proved with the asymptotic results. Finally, in Section \ref{sec:discussion}, we present some numerical computations and conclude with a summary.

\section{Some preliminary work} \label{Sec:preliminary}

\subsection{The Karlin-McGregor formula}

Introduce the notations
\begin{equation} \label{vector-notations}
\begin{split}
 & \mathbf{x}=(x_1, x_2, \cdots, x_N), \quad  \mathbf{y}=(y_1, y_2, \cdots, y_N),    \\
 & \mathbf{x}^2=(x_1^2, x_2^2, \cdots, x_N^2), \quad  |\mathbf {x}|=\sqrt{x_1^2+x_2^2+\cdots+x_N^2}.
 \end{split}
\end{equation}
and the constant vectors
\begin{equation}
\mathbf{a}=(a, a, \cdots, a), \quad \mathbf{0}= (0, 0, \cdots, 0).
\end{equation}
According to the Karlin-McGregor formula in the affine Weyl alcove of height $M$ (see  \cite{makoto2007}) and the definition of $N$ non-intersecting Bessel processes, we obtain
\begin{equation} \label{p-relation-qqm}
    \mathbb{P}(\max_{0 < t < 1} b_N(t) <M)=\lim_{\mathbf{x} \to \mathbf{a}, \mathbf{y} \to \mathbf{0}}\frac{q^M(\mathbf{x},\mathbf{y})}{q(\mathbf{x},\mathbf{y})},
\end{equation}
with
\begin{equation} \label{qm-def}
    q^M(\mathbf{x},\mathbf{y})
        =\det_{1 \le i, j \le N}[p^M(1,y_j|x_i)], \qquad q(\mathbf{x},\mathbf{y})=\det_{1 \le i, j \le N}[p(1,y_j|x_i)].
\end{equation}
Here, $p(t,y|x)$ is the transition probability defined in \eqref{p1} and \eqref{p11}, and $p^M(t,y|x)$ is the transition probability of the particle passes from $x$ to $y$ over the time interval $t$ with an absorbing wall at the position $M$.

To find $p^M(t,y|x)$, let us consider the following diffusion equation
\begin{equation}
    \begin{cases}
      \displaystyle u_t=\frac{1}{2} \left( u_{yy}+\frac{2 \alpha+1}{y}u_y \right),  \\
     \displaystyle\lim_{t \to 0}u(t,y)=\delta(y-x),  \\
     u(t,M)=u_{y}(t,0)=0;
    \end{cases}
\end{equation}
see similar equations in \cite{Abraham, borodin}.
Using the method of separation of variables, one gets the unique solution to the above equation
\begin{equation}
    u(t,y)=\sum^{\infty}_{n=1}\frac{1}{M^2J_{\alpha +1}^2(x_{n, \alpha})}\frac{1}{(xy)^{\alpha}}J_{\alpha}(\frac{x_{n,\alpha}}{M}x)J_{\alpha}(\frac{x_{n,\alpha}}{M}y)e^{-\frac{x^2_{n, \alpha}}{2M^2}t},
\end{equation}
which is the transition density with respect to the speed measure $m_{\alpha}(dy)=2y^{2 \alpha +1}dy$. Therefore, we obtain
\begin{equation}\label{p2}
    p^M(t,y|x)=\sum^{\infty}_{n=1}\frac{2}{M^2J_{\alpha +1}^2(x_{n, \alpha})}\frac{y^{\alpha +1}}{x^{\alpha}}J_{\alpha}(\frac{x_{n,\alpha}}{M}x)J_{\alpha}(\frac{x_{n,\alpha}}{M}y)e^{-\frac{x^2_{n, \alpha}}{2M^2}t}
\end{equation}
where $J_{\alpha}(z)$ is the Bessel function of the first kind
\begin{equation}\label{besselj}
    J_{\alpha}(z)=\sum^{\infty}_{k=0}(-1)^k \frac{(\frac{z}{2})^{2k+\alpha }}{k! \Gamma (k+\alpha +1)}
\end{equation}
 and $x_{n, \alpha}$'s are its zeros given in \eqref{bessel-zeros}.

\begin{remark}
When $\alpha = -\frac{1}{2}$, the zeroes in \eqref{bessel-zeros} are $x_{n,-\frac{1}{2}}=(n-\frac{1}{2})\pi$, $n \in \mathbb{N}$. Then, the  transition probability \eqref{p2} becomes
\begin{equation}\label{pp2}
    p^M(t, y|x)=\sum_{n=1}^{\infty} \frac{2}{M}\cos{\left(\frac{(n-\frac{1}{2})\pi}{M}x \right)}\cos{\left (\frac{(n-\frac{1}{2})\pi}{M}y \right)}e^{-\frac{(n-\frac{1}{2})^2\pi^2}{2M^2}t}.
\end{equation}
By the Poisson summation formula, we have
\begin{equation}
    p^M(t, y|x)=\sum^{\infty}_{n=-\infty}\frac{(-1)^n}{\sqrt{2 \pi t}}\left[e^{-\frac{(y-x-2nM)^2}{2t}}+e^{-\frac{(y+x+2nM)^2}{2t}}\right].
\end{equation}
It is easily seen from the above formula that $\lim\limits_{M \to \infty} p^M(t, y|x) = p(t, y|x)$. This, together with \eqref{p-relation-qqm}, gives us the desired result $\lim\limits_{M \to \infty}\mathbb{P}(\max\limits_{0 < t < 1} b_N(t) <M) = 1$.
\end{remark}

\subsection{Some properties for determinants}

We need some preliminary results for determinants.
\begin{lemma}\label{lemma1}
Let $f$ and $g$ be  functions with two independent variables, and $\mathbf{n}=(n_1, n_2, \cdots, n_N)$ be an $N$-dimensional vector. Then, we have
    \begin{equation} \label{one-two-det}
        \sum_{\mathbf{n} \in \mathbb{N}^N}\det_{1 \le i,j \le N}[f(x_i,n_j)g(y_j,n_j)]=\frac{1}{N!}\sum_{\mathbf{n} \in \mathbb{N}^N}\det_{1 \le i,j \le N}[f(x_i,n_j)]\det_{1 \le i,j \le N}[g(y_i,n_j)].
    \end{equation}
\end{lemma}
\begin{proof}
From the Leibniz formula for the determinant, we have
\begin{equation*}
\det_{1 \le i,j \le N}[f(x_i,n_j)g(y_j,n_j)] =
            \sum_{\sigma} \left( \sgn \sigma\prod_{j=1}^N f(x_{\sigma (j)},n_j)g(y_j,n_j) \right),
\end{equation*}
where $\sigma$ is a permutations  of the set $\{1, 2, \cdots, N\}$. Substituting the above formula into the left-hand side of \eqref{one-two-det}, we get
\begin{equation*}
\sum_{\mathbf{n} \in \mathbb{N}^N}\det_{1 \le i,j \le N}[f(x_i,n_j)g(y_j,n_j)] = \sum_{\mathbf{n} \in \mathbb{N}^N}\sum_{\sigma} \sgn \sigma\prod_{j=1}^N f(x_{\sigma (j)},n_j)g(y_j,n_j).
\end{equation*}
Let us change the index $n_j$ to $n_{\tau(j)}$, where $\tau $ is an arbitrary permutation of the set $\{1, 2, \cdots, N\}$. Since the summation is taken for all $n \in \mathbb{N}^N$, we have
\begin{equation*}
\sum_{\mathbf{n} \in \mathbb{N}^N}\det_{1 \le i,j \le N}[f(x_i,n_j)g(y_j,n_j)] = \frac{1}{N!}\sum_{\mathbf{n} \in \mathbb{N}^N}\sum_{\sigma}\sum_{\tau} \sgn \sigma\prod_{j=1}^N f(x_{\sigma (j)},n_{\tau(j)}) \prod_{j=1}^N  g(y_j,n_{\tau(j)}).
\end{equation*}
Consider a new permutation $\rho$ defined as $ \rho: = \tau \circ  \sigma^{-1}$. Denote $l = \sigma(j)$, then $\tau(j) = \tau (\sigma^{-1}(l)) = \rho(l)$. We change the index in the first product in the  above formula and obtain
\begin{equation*}
\sum_{\mathbf{n} \in \mathbb{N}^N}\det_{1 \le i,j \le N}[f(x_i,n_j)g(y_j,n_j)] =\frac{1}{N!}\sum_{\mathbf{n} \in \mathbb{N}^N}\sum_{\tau}\sum_{\rho} \sgn \tau \sgn \rho \prod_{l=1}^N f(x_l,n_{\rho(l)}) \prod_{j=1}^N   g(y_j,n_{\tau(j)}).
\end{equation*}
Separating the summations about $\tau$ and $\rho$, we get
    \begin{equation*}
        \begin{split}
            \sum_{\mathbf{n} \in \mathbb{N}^N}\det_{1 \le i,j \le N}[f(x_i,n_j)g(y_j,n_j)]
            &=\frac{1}{N!}\sum_{\mathbf{n} \in \mathbb{N}^N}\sum_{\rho}  \sgn \rho\prod_{l=1}^N f(x_l,n_{\rho(l)})\sum_{\tau}\sgn \tau \prod_{j=1}^Ng(y_j,n_{\tau(j)}).
        \end{split}
    \end{equation*}
    With the Leibniz formula for the determinants again, we obtain \eqref{one-two-det} from the above formula.

    This completes the proof of the lemma.
\end{proof}
\begin{remark}
One may take matrix transpose in \eqref{one-two-det} and rewrite the formula as
\begin{equation} \label{one-two-det-tran}
        \sum_{\mathbf{n} \in \mathbb{N}^N}\det_{1 \le i,j \le N}[f(x_j,n_i)g(y_i,n_i)]=\frac{1}{N!}\sum_{\mathbf{n} \in \mathbb{N}^N}\det_{1 \le i,j \le N}[f(x_j,n_i)]\det_{1 \le i,j \le N}[g(y_j,n_i)].
\end{equation}
Moreover, from the proof, it is easy to verify that Lemma \ref{lemma1} still holds if an additional factor $\prod_{i=1}^{N} h(n_i)$ appears in the summation. More precisely, we have
\begin{equation} \label{one-two-det-2}
    \begin{split}
        &\sum_{\mathbf{n} \in \mathbb{N}^N}\left (\prod_{i=1}^{N} h(n_i)\right )\det_{1 \le i,j \le N}[f(x_i,n_j)g(y_j,n_j)]\\
        & \qquad \qquad =\frac{1}{N!}\sum_{\mathbf{n} \in \mathbb{N}^N}\left (\prod_{i=1}^{N} h(n_i)\right )\det_{1 \le i,j \le N}[f(x_i,n_j)]\det_{1 \le i,j \le N}[g(y_i,n_j)].
    \end{split}
\end{equation}
\end{remark}

To study the limits of determinants in \eqref{p-relation-qqm}, we will prove one more lemma below. First, let us consider the following Schur function
\begin{equation} \label{schur-def}
    s_{\pmb{\mu}}(\mathbf{x})=\frac{\det \limits_{1 \le j,k \le N}[x_j^{{\mu}_k+N-k}]}{\det \limits_{1 \le j,k \le N}[x_j^{N-k}]}
\end{equation}
with $\mathbf{x}$ being the vector defined in \eqref{vector-notations}; see \cite{macdonald}.
Here $\pmb{\mu}=(\mu _1, \mu _2, \cdots, \mu _N)$ is a sequence of non-negative integers in decreasing order $\mu _1 \ge \mu _2 \ge \cdots \ge \mu _N$, which is called a partition. Obviously, both the numerator and denominator are Vandermonde determinants. For example, we have
\begin{equation}
    \det \limits_{1 \le j,k \le N}[x_j^{N-k}]=\prod_{1 \le j<k \le N}(x_j-x_k)=(-1)^{\frac{N(N-1)}{2}}\Delta (\mathbf x).
\end{equation}
Note that $s_{\pmb{\mu}}(\mathbf{0})=0$ unless $\pmb{\mu} =\mathbf 0$. Moreover, we have the following property
\begin{equation} \label{schur-limit}
      s_{\mathbf{0}}(\mathbf{x}) \equiv 1 \qquad \textrm{and} \qquad s_{\pmb{\mu}}(\mathbf{x}) = O(|\mathbf{x}|) \ \ \textrm{as } \mathbf{x} \to \mathbf{0} \ \ \textrm{if } \pmb{\mu} \ne \mathbf{0}.
      %\quad =\begin{cases}
%      1, & \pmb{\mu}=\mathbf{0} \\
%      O(|\mathbf{x}|) , & \pmb{\mu} \ne \mathbf{0}
%      \end{cases}
\end{equation}
Then, we have the following result.

\begin{lemma} \label{lemma2}
    Let  $f$ be a smooth function, then we have
    \begin{equation} \label{fxy-limit}
        \det_{1 \le i,j \le N}[f(x_i y_j)]=\left(\prod_{j=1}^N\frac{f^{(j-1)}(ay_{N+1-j})}{\Gamma(j)} \right)\Delta (\mathbf{x-a}) \Delta (\mathbf y)\{1+O(|\mathbf{x-a}|)\}
    \end{equation}
    as $\mathbf{x} \to \mathbf{a}$. %For $y \to b$, we have a similar equation.
\end{lemma}

\begin{proof}
   We expand $f(x_i y_j)$ at $x_i=a$ into  a Taylor series
   \begin{equation*}
    f(x_i y_j) = \sum_{k=0}^{\infty}\frac{f^{(k)}(ay_{j})}{k!}(x_i-a)^k y_j^k.
   \end{equation*}
    With \eqref{one-two-det} and the above formula, we have
    \begin{equation*}
        \begin{split}
            \det_{1 \le i,j \le N}[f(x_i y_j)]
            &=\sum_{\mathbf k \in \mathbb{N}_0^N}\det_{1 \le i,j \le N} \left[\frac{f^{(k_j)}(ay_{j})}{k_j!}(x_i-a)^{k_j} y_j^{k_j} \right] \\
            &=\frac{1}{N!}\sum_{\mathbf k \in \mathbb{N}_0^N}\prod_{j=1}^N\frac{f^{(k_j)}(ay_{j})}{k_j!}\det_{1 \le i,j \le N}[(x_i-a)^{k_j}]\det_{1 \le i,j \le N}[y_i^{k_j}],
            \end{split}
    \end{equation*}
    where $\mathbb{N}_0=\{0, 1, 2, \cdots\}$.
    We order the index $\mathbf k = (k_1, k_2, \cdots, k_N)$ such that $\{ k_j \}_{j=1}^N$ is a decreasing sequence. This gives us
    \begin{equation*}
            \det_{1 \le i,j \le N}[f(x_i y_j)]
              =\sum_{0 \le k_N < k_{N-1} <\cdots <k_1}\prod_{j=1}^N\frac{f^{(k_j)}(ay_{j})}{k_j!}\det_{1 \le i,j \le N}[(x_i-a)^{k_j}]\det_{1 \le i,j \le N}[y_i^{k_j}].
    \end{equation*}
    Changing the index from $\mathbf k$ to $\pmb\mu$ with $\mu_j=k_j-N+j$, we get from the definition of the Schur function in \eqref{schur-def} that
    \begin{equation*}
        \begin{split}
            \det_{1 \le i,j \le N}[f(x_i y_j)] &=\sum_{\pmb{\mu}}\prod_{j=1}^N\frac{f^{(\mu_j+N-j)}(ay_{j})}{(\mu_j+N-j)!}s_{\pmb{\mu}}(\mathbf x-\mathbf a)\Delta(\mathbf{x-a})s_{\pmb{\mu}}(\mathbf y)\Delta(\mathbf y). %\\
            %&=(\prod_{j=1}^N\frac{f^{(N-j)}(ay_{j})}{(N-j)!})\Delta (\mathbf{x-a})\Delta (\mathbf y)\times\{1+O(|\mathbf{x-a}|)\}\\
            %&=(\prod_{j=1}^N\frac{f^{(j-1)}(ay_{N+1-j})}{\Gamma(j)})\Delta (\mathbf{x-a})\Delta (\mathbf y)\times\{1+O(|\mathbf{x-a}|)\}
        \end{split}
    \end{equation*}
    Finally, using the asymptotics of the Schur function in \eqref{schur-limit}, we obtain \eqref{fxy-limit}.
\end{proof}

\section{Derivation of the distribution function}\label{sec:derivation}

In this section, we will prove our main theorem by  studying   some properties of $ q(\mathbf{x}, \mathbf{y})$ and $ q^M(\mathbf{x}, \mathbf{y})$ defined in \eqref{qm-def}.

\subsection{Asymptotics of $q(\mathbf{x}, \mathbf{y})$ and $q^M(\mathbf{x}, \mathbf{y})$}

From the transition probabilities $p(t,y|x)$ in \eqref{p1} and $p^M(t,y|x)$ in \eqref{p2}, we have
\begin{eqnarray}
  q(\mathbf{x}, \mathbf{y})
        &=&\det_{1 \le i, j \le N} \left[\frac{y_j^{\alpha+1}}{x_i^{\alpha}}e^{-\frac{x_i^2+y_j^2}{2}}I_{\alpha}(x_i y_j) \right], \label{q-det-formula} \\
        q^M(\mathbf{x}, \mathbf{y})&=&\det_{1 \le i, j \le N}\left[\sum^{\infty}_{n=1}\frac{2}{M^2J_{\alpha +1}^2(x_{n, \alpha})}\frac{y_j^{\alpha +1}}{x_i^{\alpha}}J_{\alpha}(\frac{x_{n,\alpha}}{M}x_i)J_{\alpha}(\frac{x_{n,\alpha}}{M}y_j)e^{-\frac{x^2_{n, \alpha}}{2M^2}}\right]. \label{qm-det-formula}
\end{eqnarray}
To obtain the probability of the maximal height $\D \mathbb{P}(\max_{0 < t < 1} b_N(t) <M)$ in \eqref{p-relation-qqm}, we need the asymptotics of the above two functions as $\mathbf{x} \to \mathbf{a}$ and $\mathbf{y} \to \mathbf{0}$.

\begin{prop}
 For $a>0, N \ge 1$, $M \ge a$, as $\mathbf{x} \to \mathbf{a}$ and $\mathbf{y} \to \mathbf{0}$, we have
\begin{equation} \label{q-limit-1}
\begin{split}
    q(\mathbf{x}, \mathbf{y}) &
    = \frac{a^{\frac{N(N-1)}{2}} e^{-\frac{a^2 N}{2} } }{2^{\frac{N(N-1+2\alpha)}{2}}} \left (\prod_{k=1}^N\frac{y_k^{2 \alpha +1} }{\Gamma(k)\Gamma(\alpha+k)} \right ) \\
     & \qquad \qquad \times\Delta(\mathbf{x-a})\Delta(\mathbf{y}^2)\times  \{1+O(|\mathbf{x-a}|)+O(|\mathbf{y}|)\}
    \end{split}
\end{equation}
and
\begin{equation} \label{qm-limit-1}
    \begin{split}
        q^M(\mathbf{x}, \mathbf{y})
        &=\frac{1}{N!2^{N(N+\alpha -2)}M^{\frac{N(3N+2\alpha+1)}{2}}}\left(\prod_{k=1}^N \frac{y_k^{2 \alpha +1}}{x_k^{\alpha}(\Gamma (k))^2\Gamma (k+\alpha)}\right)\Delta(\mathbf {x-a})\Delta(\mathbf {y^2})\\
        & \quad \times \left (\sum_{\mathbf n \in \mathbf N^N} \left (\prod_{s=1}^N \frac{(-1)^{s-1}x_{n_s, \alpha}^{ \alpha}J_{\alpha}^{(s-1)}(\frac{a}{M}x_{n_s,\alpha})e^{-\frac{x_{n_s, \alpha}^2}{2M^2}}}{J_{\alpha +1}^2(x_{n_s, \alpha})} \right ) \Delta (\mathbf{x_{n, \alpha}})\Delta(\mathbf{x_{n, \alpha}^2}) \right )\\
        & \quad \times \{1+O(|\mathbf {x-a}|)+O(|\mathbf y|)\},
    \end{split}
\end{equation}
where $\Delta(\mathbf{x})$ is defined in \eqref{vandermond} and
\begin{equation}
  \mathbf{x_{n, \alpha}} = (x_{n_1, \alpha}, x_{n_2, \alpha}, \cdots , x_{n_N, \alpha}) .
\end{equation}
When $a=0$, we have
\begin{equation} \label{q-limit-2}
    q(\mathbf{x}, \mathbf{y})= \frac{1}{2^{N(N-1+\alpha)}} \left (\prod_{k=1}^N\frac{y_k^{2 \alpha +1}  }{\Gamma(k)\Gamma(\alpha+k)} \right )\Delta(\mathbf{x}^2)\Delta(\mathbf{y}^2)\times \{1+O(|\mathbf{x}|)+O(|\mathbf{y}|)\}
\end{equation}
and
\begin{equation} \label{qm-limit-2}
    \begin{split}
        q^M(\mathbf{x}, \mathbf{y})
        &=\frac{1}{N! 2^{N(2N+2\alpha -3)} M^{2N(N+\alpha)}} \left(\prod_{k=1}^N \frac{y_k^{2 \alpha +1}}{(\Gamma (k))^2 \Gamma (k+\alpha)^2}\right)\Delta(\mathbf {x}^2)\Delta(\mathbf {y^2})\\
        &\quad \times \left (\sum_{\mathbf n \in \mathbf N^N} \left (\prod_{s=1}^N \frac{x_{n_s, \alpha}^{2 \alpha}e^{-\frac{x_{n_s, \alpha}^2}{2M^2}}}{J_{\alpha +1}^2(x_{n_s, \alpha})} \right ) \left (\Delta(\mathbf{x_{n, \alpha}^2})\right )^2 \right ) \times \{1+O(|\mathbf {x}|)+O(|\mathbf y|)\}.
    \end{split}
\end{equation}
\end{prop}

\begin{proof}
    By the multilinearity of the determinant, we have from \eqref{q-det-formula}:
    \begin{equation} \label{q-det-formula-2}
        q(\mathbf{x}, \mathbf{y})
       =\left(\prod_{k=1}^N\frac{y_k^{\alpha+1}}{x_k^{\alpha}}e^{-\frac{x_k^2+y_k^2}{2}} \right)\det_{1 \le i, j \le N}[I_{\alpha}(x_i y_j)].
    \end{equation}
    With the definition of $I_{\alpha}(z)$ in \eqref{defI}, we get
    \begin{equation*}
    \begin{split}
        \det_{1 \le i, j \le N}[I_{\alpha}(x_i y_j)] & =\det_{1 \le i, j \le N}\left [\sum_{k=0}^{\infty} \left (\frac{1}{2} \right )^{2k+\alpha}\frac{(x_i y_j)^{2k+\alpha}}{k! \, \Gamma (k+\alpha+1)}\right ] \\ &
         = \frac{\prod \limits_{k=1}^N(x_k y_k)^{\alpha}  }{2^{\alpha N} }   \det_{1 \le i, j \le N}\left [\sum_{k=0}^{\infty} \left (\frac{1}{2} \right )^{2k}\frac{(x_i y_j)^{2k}}{k! \, \Gamma (k+\alpha+1)}\right ] .
        \end{split}
    \end{equation*}
    Moreover, from Lemma \ref{lemma1}, we obtain the following expression
    \begin{equation*}
            \det_{1 \le i, j \le N}[I_{\alpha}(x_i y_j)]=
             \frac{ \prod \limits_{k=1}^N(x_k y_k)^{\alpha}  }{2^{\alpha N} N!} \sum_{\mathbf k \in \mathbb{N}_0^N}\frac{ 2^{ -\sum \limits_{j=1}^N 2 k_j}}{\prod\limits_{j=1}^N k_j! \, \Gamma(k_j+\alpha+1)} \det_{1 \le i,j \le N}[x_i^{2k_j}]\det_{1 \le i,j \le N}[y_i^{2k_j}] .
    \end{equation*}
    Since $\displaystyle \det_{1 \le i,j \le N}[x_i^{2k_j}]=0$ and $\displaystyle \det_{1 \le i,j \le N}[y_i^{2k_j}]=0$ for any $k_i = k_j$, we may order the index $\mathbf{k}=(k_1, k_2, \cdots, k_N)$ such that $k_1 > k_2 > \cdots >k_N$.
    %Then, the above determinant reduces to
    %\begin{equation*}
 %            \det_{1 \le i, j \le N}[I_{\alpha}(x_i y_j)]= \frac{ \prod \limits_{k=1}^N(x_k y_k)^{\alpha}  }{2^{\alpha N} }  \sum_{0 \le k_N < k_{N-1} <\cdots <k_1}\frac{ 2^{-\sum \limits_{j=1}^N 2 k_j}}{\prod\limits_{j=1}^N k_j! \, \Gamma(k_j+\alpha+1)} \det_{1 \le i,j \le N}[x_i^{2k_j}]\det_{1 \le i,j \le N}[y_i^{2k_j}].
%    \end{equation*}
    Moreover, set $\mu _j= k_j-N+j$. Obviously, we have $\sum\limits_{j=1}^N k_j=\sum\limits_{j=1}^N \mu_j+N(N-1)/2$. Replacing the indexes  $k_j$ by $\mu_j$ in the above formula, we obtain from the above formula and the definition of the Schur function in \eqref{schur-def}:
    \begin{equation}  \label{I-alpha-det}
            \det_{1 \le i, j \le N}[I_{\alpha}(x_i y_j)]= \frac{ \prod \limits_{k=1}^N(x_k y_k)^{\alpha}  }{2^{\alpha N} }   \Delta(\mathbf{x}^2)\Delta(\mathbf{y}^2)\sum_{\pmb{\mu}}\frac{ 2^{ -\sum\limits_{j=1}^N 2\mu_j-N(N-1)}s_{\pmb{\mu}}(\mathbf x^2)s_{\pmb{\mu}}(\mathbf y^2)}{\prod \limits_{j=1}^N(\mu_{N-j+1}+j-1)! \, \Gamma(\mu_{N-j+1}+\alpha+j)}.
    \end{equation}
    Recall that $s_{\pmb{\mu}}(\mathbf{0})=0$ unless $\pmb{\mu} =\mathbf 0$. As $\mathbf{y} \to \mathbf{0}$, the leading term in the above summation comes from the index $\pmb{\mu} = \mathbf{0}$. Moreover, since $s_{\pmb{\mu}}(\mathbf x^2)$ and $s_{\pmb{\mu}}(\mathbf y^2)$ always choose the same index $\pmb{\mu}$, we have
    \begin{equation} \label{det-schur-limit}
      \sum_{\pmb{\mu}}\frac{ 2^{-\sum\limits_{j=1}^N 2\mu_j-N(N-1)}s_{\pmb{\mu}}(\mathbf x^2)s_{\pmb{\mu}}(\mathbf y^2)}{\prod \limits_{j=1}^N(\mu_{N-j+1}+j-1)! \, \Gamma(\mu_{N-j+1}+\alpha+j)} = \frac{2^{- N(N-1)} }{\prod \limits_{j=1}^N(j-1)! \, \Gamma(\alpha+j)} \{1+O(|\mathbf{y}|) \},
    \end{equation}
    as $\mathbf{y} \to \mathbf{0}$ uniformly for all $\mathbf{x}$. Since
     \begin{equation}
       \Delta(\mathbf{x}^2) \sim (2a)^{\frac{N(N-1)}{2}} \Delta(\mathbf{x}-\mathbf{a}) \{ 1 + O(|\mathbf{x} - \mathbf{a}|)\}, \qquad \textrm{as } \mathbf{x} \to \mathbf{a} \quad \textrm{if} \quad a > 0,
     \end{equation}
     we obtain \eqref{q-limit-1} for $a>0$ and \eqref{q-limit-2} for $a=0$, respectively, by using \eqref{q-det-formula-2} and the above three formulae.

    Next, we consider $q^M(\mathbf{x}, \mathbf{y})$ in \eqref{qm-det-formula}.  Similar to the above computations, we obtain from  the multilinearity of determinants and Lemma \ref{lemma1},
    \begin{eqnarray}
      && q^M(\mathbf{x}, \mathbf{y}) = \frac{2^N \prod \limits_{k=1}^N \frac{y_k^{\alpha+1}}{x_k^{\alpha}} }{N! \, M^{2N}} \nonumber \\
            %=\frac{2^N \prod \limits_{k=1}^N \frac{y_k^{\alpha+1}}{x_k^{\alpha}} }{M^{2N}} \sum_{\mathbf n \in \mathbb{N}^N}\left(\prod_{s=1}^N \frac{e^{-\frac{x_{n_s, \alpha}^2}{2M^2}}}{J_{\alpha +1}^2(x_{n_s, \alpha})}\right)\det_{1 \le i, j \le N}\left[J_{\alpha}(\frac{x_{n_j,\alpha}}{M}x_i)J_{\alpha}(\frac{x_{n_j,\alpha}}{M}y_j)\right] \nonumber \\
            && \qquad \times \sum_{\mathbf n \in \mathbb{N}^N}\left(\prod_{s=1}^N \frac{e^{-\frac{x_{n_s, \alpha}^2}{2M^2}}}{J_{\alpha +1}^2(x_{n_s, \alpha})}\right) \det_{1 \le i, j \le N}\left[J_{\alpha}(\frac{x_{n_j,\alpha}}{M}x_i)\right]\det_{1 \le i, j \le N}\left[J_{\alpha}(\frac{x_{n_j,\alpha}}{M}y_i)\right] . \label{qm-det-formula-two}
    \end{eqnarray}
    Regarding the determinant $\displaystyle \det_{1 \le i, j \le N}[J_{\alpha}(\frac{x_{n_j,\alpha}}{M}x_i)]$, we first expand $J_{\alpha}(\frac{x_{n_j,\alpha}}{M}x_i)$ into a Taylor series at $x_i = a$. Then, from Lemma \ref{lemma2}, we have
    \begin{equation} \label{ja-det-a-asy}
        \begin{split}
            \det_{1 \le i, j \le N}\left[J_{\alpha}(\frac{x_{n_j,\alpha}}{M}x_i)\right]
            &=\det_{1 \le i, j \le N}\left [\sum_{m=0}^{\infty} \frac{J_{\alpha}^{(m)}(\frac{a}{M}x_{n_j,\alpha})}{m!}(\frac{x_{n_j,\alpha}}{M})^m(x_i-a)^m\right]\\
            %&=\sum_{\mathbf{m}\in \mathbb{N}^N}\left(\prod_{j=1}^N\frac{J_{\alpha}^{(m_j)}(\frac{a}{M}x_{n_j,\alpha})}{m_j!M^{m_j}}\right)\det_{1 \le i, j \le N}[((x_i-a)x_{n_j, \alpha})^{m_j}]\\
            %&=\sum_{\mathbf{m}\in \mathbb{N}^N}\left(\prod_{j=1}^N\frac{J_{\alpha}^{(m_j)}(\frac{a}{M}x_{n_j,\alpha})}{m_j!M^{m_j}}\right)\frac{1}{N!}\det_{1 \le i, j \le N}[(x_i-a)^{m_j}]\det_{1 \le i, j \le N}[x_{n_i, \alpha}^{m_j}]\\
            %&=\sum_{0\le m_N<m_{N-1}<\cdots<m_1}(\prod_{j=1}^N\frac{J_{\alpha}^{(m_j)}(\frac{a}{M}x_{n_j,\alpha})}{m_j!M^{m_j}})\det_{1 \le i, j \le N}[(x_i-a)^{m_j}]\det_{1 \le i, j \le N}[x_{n_i, \alpha}^{m_j}]
            &=\left(\prod_{j=1}^N\frac{J_{\alpha}^{(j-1)}(\frac{a}{M}x_{n_j,\alpha})}{(j-1)!M^{j-1}}\right)\Delta(\mathbf {x-a})\Delta(\mathbf{x_{n, \alpha}})\times\{1+O(|\mathbf{x-a}|)\},
        \end{split}
    \end{equation}
    as $\mathbf{x} \to \mathbf{a}$.   If the starting point $a=0$, the expansion near $x = a$ is different as $J_\alpha(x)$ has an algebraic singularity at the origin; cf. \eqref{besselj}. In this case, we adopt similar computations in deriving \eqref{I-alpha-det} to obtain
    \begin{eqnarray}
            && \displaystyle \det_{1 \le i, j \le N}\left[J_{\alpha}(\frac{x_{n_j,\alpha}}{M}x_i)\right] =\det_{1 \le i, j \le N}\left[\sum_{m=0}^{\infty}(-1)^m\frac{(\frac{x_{n_j, \alpha}}{2M}x_i)^{2m+\alpha}}{m! \Gamma (m+\alpha+1)}\right] \nonumber \\
            && \displaystyle \qquad \qquad = \frac{\prod\limits_{k=1}^N (x_{n_k, \alpha}x_k)^{\alpha}}{(2M)^{\alpha N}}
             \Delta(\mathbf{x}^2) \Delta (\mathbf{x^{\rm 2}_{n, \alpha}})  \nonumber \\
            && \qquad \qquad \qquad \times  \sum_{\pmb{\mu}}\frac{ (-1)^{\sum\limits_{j=1}^N \mu_j+N(N-1)/2} (2M)^{ -\sum\limits_{j=1}^N 2\mu_j-N(N-1)}s_{\pmb{\mu}}(\mathbf x^2)s_{\pmb{\mu}}(\mathbf{x^{\rm 2}_{n, \alpha}} )}{\prod \limits_{j=1}^N(\mu_{N-j+1}+j-1)! \, \Gamma(\mu_{N-j+1}+\alpha+j)}.
    \end{eqnarray}
  Similar to \eqref{det-schur-limit}, the above formula yields
  \begin{equation} \label{ja-det-0-asy}
  \det_{1 \le i, j \le N}\left[J_{\alpha}(\frac{x_{n_j,\alpha}}{M}x_i)\right] = \frac{\prod\limits_{k=1}^N (x_{n_k, \alpha}x_k)^{\alpha}}{(2M)^{N ( N -1 + \alpha ) }}
             \Delta(\mathbf{x}^2) \Delta (\mathbf{x^{\rm 2}_{n, \alpha}})   \prod \limits_{j=1}^N \frac{ (-1)^{j-1}  }{(j-1)! \, \Gamma(\alpha+j)} \{1+O(|\mathbf{x}|) \},
  \end{equation}
  as $\mathbf{x} \to \mathbf{0}$.
  Of course, similar formula also holds for $\det\limits_{1 \le i, j \le N}\left[J_{\alpha}(\frac{x_{n_j,\alpha}}{M}y_i)\right]$ when $\mathbf{y} \to \mathbf{0}$.

  Finally, with \eqref{qm-det-formula-two}, \eqref{ja-det-a-asy} and \eqref{ja-det-0-asy}, we obtain the desired asymptotics in \eqref{qm-limit-1} for $a>0$ and \eqref{qm-limit-2} for $a =0$, respectively.
\end{proof}

\subsection{Proof of Theorem \ref{the:theorem}}

The asymptotics for $q(\mathbf{x}, \mathbf{y})$ and $q^M(\mathbf{x}, \mathbf{y})$ in the previous are enough for us to derive the probability for the maximum height.

  Recall the expression of $\D \mathbb{P}(\max_{0 < t < 1} b_N(t) <M)$ in terms of $q(\mathbf{x},\mathbf{y})$ and $q^M(\mathbf{x},\mathbf{y})$ in \eqref{p-relation-qqm}. When $a>0$, using the asymptotics obtained in \eqref{q-limit-1} and \eqref{qm-limit-1}, we have
    \begin{equation} \label{distribution}
        \begin{split}
            \mathbb{P}(\max_{0 < t < 1} b_N(t) <M)
            &= c_N(\alpha) M^{-\frac{N(3N+2\alpha+1)}{2}}\frac{1}{N!}\\
            & \hspace{-1cm} \times \sum_{\mathbf n \in \mathbb{N}^N} \left [ \left( \prod_{s=1}^N \frac{(-1)^{s-1}x_{n_s, \alpha}^{ \alpha}J_{\alpha}^{(s-1)}(\frac{a}{M}x_{n_s,\alpha})e^{-\frac{x_{n_s, \alpha}^2}{2M^2}}}{J_{\alpha +1}^2(x_{n_s, \alpha})} \right)  \Delta(\mathbf{x_{n, \alpha}})\Delta(\mathbf{x_{n, \alpha}^2})\right],
        \end{split}
    \end{equation}
    where $c_N(\alpha)$ is given in \eqref{cn-def}.   Now, let us focus on the summation in the above formula and put it into a determinantal form. From the definition of Vandermonde determinants, we have
    \begin{equation*}
     \Delta(\mathbf{x_{n, \alpha}})\Delta(\mathbf{x_{n, \alpha}^2}) = \det_{1 \le i, j \le N}[x_{n_i, \alpha}^{j-1}]\det_{1 \le i, j \le N}[x_{n_i, \alpha}^{2(j-1)}].
    \end{equation*}
   Using \eqref{one-two-det-2}, we get
    \begin{equation*}
        \begin{split}
           & \sum_{\mathbf n \in \mathbb{N}^N} \left(\prod_{s=1}^N \frac{(-1)^{s-1}x_{n_s, \alpha}^{ \alpha}J_{\alpha}^{(s-1)}(\frac{a}{M}x_{n_s,\alpha})e^{-\frac{x_{n_s, \alpha}^2}{2M^2}}}{J_{\alpha +1}^2(x_{n_s, \alpha})} \right) \Delta(\mathbf{x_{n, \alpha}})\Delta(\mathbf{x_{n, \alpha}^2}) \\
            & \qquad = N! \ \sum_{\mathbf n \in \mathbb{N}^N} \left(\prod_{s=1}^N \frac{(-1)^{s-1}x_{n_s, \alpha}^{ \alpha}J_{\alpha}^{(s-1)}(\frac{a}{M}x_{n_s,\alpha})e^{-\frac{x_{n_s, \alpha}^2}{2M^2}}}{J_{\alpha +1}^2(x_{n_s, \alpha})} \right)  \det_{1 \le i, j \le N}\left[x_{n_i, \alpha}^{i+2j-3}\right].
        \end{split}
    \end{equation*}
    Moving the preceding factors into the determinant and using the multilinearity of determinant, we obtain \eqref{deter} from the above formula.

     When $a=0$, from \eqref{q-limit-2} and \eqref{qm-limit-2}, the distribution formula can be written as
    \begin{equation} \label{distribution-2}
        \mathbb{P}(\max_{0 < t < 1} b_N(t) <M)=\widetilde{c}_N (\alpha) M^{-2N(N+\alpha)}\frac{1}{N!}\sum_{\mathbf n \in \mathbb{N}^N} \left [\left (\prod_{s=1}^N \frac{x_{n_s, \alpha}^{2\alpha} e^{-\frac{x_{n_s, \alpha}^2}{2M^2}}}{J_{\alpha +1}^2(x_{n_s, \alpha})}\right ) (\Delta(\mathbf{x_{n, \alpha}^2}))^2 \right ],
    \end{equation}
    where $\widetilde{c}_N (\alpha)$ is given in \eqref{cn2-def}. With the aid of \eqref{one-two-det-2} again, we obtain \eqref{deter-a=0} from above formula by similar computations.

    This completes the proof of Theorem \ref{the:theorem}.
    \hfill \qed

\subsection{Proof of Theorem \ref{thm-op}}

Before proving Theorem \ref{thm-op}, we need one more property for the derivatives of the Bessel function $J_\alpha(z)$.
\begin{lemma}
For $j \in \mathbb{N}$, we have the following differential relations
    \begin{align}
        z^{2j-1}J_{\alpha}^{(2j-1)}(z)&=[(-1)^j z^{2j-1}+P_{2j-3}(z)]J_{\alpha+1}(z)+Q_{2j-2}(z)J_{\alpha}(z), \label{bessel-derivative-1} \\
        z^{2j}J_{\alpha}^{(2j)}(z)&=\widetilde P_{2j-1}(z)J_{\alpha+1}(z)+[(-1)^j z^{2j}+\widetilde Q_{2j-2}(z)]J_{\alpha}(z), \label{bessel-derivative-2}
    \end{align}
    where $P_j(z)$, $Q_j(z)$, $\widetilde P_j(z)$ and $\widetilde Q_j(z)$ are polynomials of degree $j$. Moreover, $P_{2j-1}(z)$ and $\widetilde P_{2j-1}(z)$ are odd functions, while $Q_{2j}(z)$ and $\widetilde Q_{2j}(z)$ are even functions.
\end{lemma}

\begin{proof}
We prove this lemma  by induction. It is well-known that the Bessel functions satisfy the following recurrence relations
    \begin{align}\label{dr2}
        zJ_{\alpha}'(z)=zJ_{\alpha-1}(z)-\alpha J_{\alpha}(z), \qquad
        zJ_{\alpha}'(z)=-zJ_{\alpha+1}(z)+ \alpha J_{\alpha}(z);
    \end{align}
    see \cite[Eq. (10.6.2)]{dlmf}. From the above formulae, it is easy to see
    \begin{equation}
   z^2J_{\alpha}''(z)=z J_{\alpha+1}(z)+(-z^2 + \alpha^2 - \alpha)J_{\alpha}(z).
    \end{equation}
     The above formulae show that the lemma holds for $j=1$.

    Assume the lemma holds for  $j=k$. Then, we have
    \begin{equation*}\label{dr1}
        z^{2k}J_{\alpha}^{(2k)}(z)=\widetilde P_{2k-1}(z)J_{\alpha+1}(z)+[(-1)^k z^{2k}+\widetilde Q_{2k-2}(z)]J_{\alpha}(z).
    \end{equation*}
    Taking derivative on both sides of the above formula and using \eqref{dr2} again, we obtain
    \begin{equation}\label{dr3}
            z^{2k+1}J_{\alpha}^{(2k+1)}(z)=[(-1)^{k+1}z^{2k+1}+P_{2k-1}(z)]J_{\alpha+1}(z)+Q_{2k}(z)J_{\alpha}(z)
    \end{equation}
    with
    \begin{align*}
        &P_{2k-1}(z)=-(2k+\alpha+1)\widetilde P_{2k-1}(z)+z\widetilde P'_{2k-1}(z)-z\widetilde Q_{2k-2}(z),\\
        &Q_{2k}(z)=\alpha (-1)^{k}z^{2k}+z\widetilde P_{2k-1}(z)+(\alpha -2k)\widetilde Q_{2k-2}(z)+z\widetilde Q'_{2k-2}(z).
    \end{align*}
    So, \eqref{bessel-derivative-1} holds for $j = k+1$.  Differentiating \eqref{dr3} one more time, we also get \eqref{bessel-derivative-2} for $j = k+1$ through similar computations.

    This finish the proof of the lemma.
\end{proof}

With the above preparation, we prove our second theorem below.

\medskip

\noindent \emph{Proof of Theorem \ref{thm-op}.} Let us focus on the determinant in \eqref{deter}. The entries in the first row are given by
\begin{equation} \label{hankl-row1}
     \sum_{n=1}^{\infty}\frac{x_{n, \alpha}^{2j+ \alpha-2}e^{-\frac{x_{n, \alpha}^2}{2M^2}}}{J_{\alpha+1}^2(x_{n, \alpha})}J_{\alpha}(\frac{a}{M}x_{n,\alpha}) .
 \end{equation}
For the second row, by using \eqref{dr2}, we replace $J_{\alpha}'(\frac{a}{M}x_{n,\alpha})$ with $J_{\alpha+1}(\frac{a}{M}x_{n,\alpha}) $ and $J_{\alpha}(\frac{a}{M}x_{n,\alpha})$. This gives us
\begin{equation*}
     \sum_{n=1}^{\infty}\frac{x_{n, \alpha}^{2j+ \alpha-2}e^{-\frac{x_{n, \alpha}^2}{2M^2}}}{J_{\alpha+1}^2(x_{n, \alpha})}\left[x_{n, \alpha}J_{\alpha+1}(\frac{a}{M}x_{n,\alpha})-\frac{\alpha M}{a}J_{\alpha}(\frac{a}{M}x_{n,\alpha}) \right]  .
 \end{equation*}
 Note that the coefficient of $J_\alpha$ term is independent of $n$. Applying a row operation, we eliminate the $J_\alpha$ term in the second row and obtain
 \begin{equation} \label{hankl-row2}
     \sum_{n=1}^{\infty}\frac{x_{n, \alpha}^{2j+ \alpha-2}e^{-\frac{x_{n, \alpha}^2}{2M^2}}}{J_{\alpha+1}^2(x_{n, \alpha})}x_{n, \alpha}J_{\alpha+1}(\frac{a}{M}x_{n,\alpha}).
 \end{equation}
Regarding the third row, we use \eqref{bessel-derivative-2} to replace $x_{n,\alpha}^2 J_\alpha''(\frac{a}{M}x_{n,\alpha})$   and obtain
 \begin{equation*}
     \sum_{n=1}^{\infty}\frac{x_{n, \alpha}^{2j+ \alpha-2}e^{-\frac{x_{n, \alpha}^2}{2M^2}}}{J_{\alpha+1}^2(x_{n, \alpha})}\left[ d_1 x_{n, \alpha}J_{\alpha+1}(\frac{a}{M}x_{n,\alpha}) + \left( -x_{n, \alpha}^2 + d_0 \right) J_{\alpha}(\frac{a}{M}x_{n,\alpha})\right].
 \end{equation*}
Note that $d_0$ and $d_1$ in the above formulae are two $n$-independent constants. Then, with the first two rows given in \eqref{hankl-row1} and \eqref{hankl-row2}, we again apply row operations to simplify the entries in the third row as
 \begin{equation}
     \sum_{n=1}^{\infty}\frac{x_{n, \alpha}^{2j+ \alpha-2}e^{-\frac{x_{n, \alpha}^2}{2M^2}}}{J_{\alpha+1}^2(x_{n, \alpha})} \left[-x_{n, \alpha}^2J_{\alpha}(\frac{a}{M}x_{n,\alpha}) \right] .
 \end{equation}
One may repeat the above procedures by first replacing the $(-1)^{i-1} x_{n,\alpha}^{i-1} J_\alpha^{(i-1)}(\frac{a}{M}x_{n,\alpha})$ term with $J_{\alpha+1}$ and $J_{\alpha}$ functions, then apply row operations to simplify the entries. Here, we also need the properties that the polynomials $P_{2k-1}$, $\widetilde P_{2k-1}$ and $\widetilde Q_{2k}$, $Q_{2k}$  in \eqref{bessel-derivative-1} and \eqref{bessel-derivative-2} are even or odd functions. Therefore, the determinant in \eqref{deter} can be put into a form as follows:
\begin{equation} \label{moment}
    \det_{1 \le i, j \le N} \left[\sum_{n=1}^{\infty}\frac{(-1)^{i-1}x_{n, \alpha}^{i+2j+\alpha-3}J_{\alpha}^{(i-1)}(\frac{a}{M}x_{n,\alpha})e^{-\frac{x_{n, \alpha}^2}{2M^2}}}{J_{\alpha+1}^2(x_{n, \alpha})} \right] =\det_{1 \le i, j \le N} \left[\sum_{n=1}^{\infty}x_{n, \alpha}^{2j + \alpha-2}\frac{h_i(x_{n, \alpha}) e^{-\frac{x_{n, \alpha}^2}{2M^2}}}{J_{\alpha+1}^2(x_{n, \alpha})} \right],
\end{equation}
% For the forth row, we have
% \begin{equation}
%     \left [\sum_{n=1}^{\infty}\frac{x_{n, \alpha}^{2j-2+\alpha}e^{-\frac{x_{n, \alpha}^2}{2M^2}}}{J_{\alpha+1}^2(x_{n, \alpha})}[-x_{n, \alpha}^3J_{\alpha+1}(\frac{a}{M}x_{n,\alpha})+(d'_2x_{n, \alpha}^2+d'_0)J_{\alpha}(\frac{a}{M}x_{n,\alpha})+d'_1x_{n, \alpha}J_{\alpha+1}(\frac{a}{M}x_{n,\alpha})] \right ]_{j=1}^N
% \end{equation}
% where $d'_0$, $d'_1$ and $d'_2$ are constants related to $\alpha, M, a$. After row elimination, we simplify the forth row as
% \begin{equation}
%     \left [\sum_{n=1}^{\infty}\frac{x_{n, \alpha}^{2j-2+\alpha}e^{-\frac{x_{n, \alpha}^2}{2M^2}}}{J_{\alpha+1}^2(x_{n, \alpha})}[-x_{n, \alpha}^3J_{\alpha+1}(\frac{a}{M}x_{n,\alpha})] \right ]_{j=1}^N
% \end{equation}
% Repeating the above procedures several times, we obtain
where
\begin{equation*}
    h_i(x)=  \begin{cases}
        (-1)^{\frac{i-1}{2}} x^{i-1} J_{\alpha}(\frac{a}{M}x),& \mbox{if } i \mbox{ is odd}, \\
        (-1)^{\frac{i-2}{2}} x^{i-1} J_{\alpha+1}(\frac{a}{M}x),& \mbox{if } i \mbox{ is even}.
        \end{cases}
\end{equation*}
Note that the entries in the right-hand side of \eqref{moment} are indeed the moments $m_{k}^{(i)}, \ i=1,2,$ defined in \eqref{mk-def}.
Finally, we obtain \eqref{hankel} by  interchanging rows in \eqref{moment}, such that all the ones containing $J_{\alpha}(\frac{a}{M}x_{n,\alpha})$ are moved to the top block, while the ones with  $J_{\alpha+1}(\frac{a}{M}x_{n,\alpha})$ are moved to the bottom block.

For the case $a=0$, it is straightforward to see that the entries in the determinant of \eqref{deter-a=0} are entries associated with the discrete orthogonal polynomials defined in \eqref{mk-def-2}. Then, the formula \eqref{hankel-a=0} follows immediately.

This completes the proof of Theorem \ref{thm-op}.
\hfill \qed

\section{Numerical simulations and conclusions}\label{sec:discussion}

\subsection{Numerical simulations}

Based on the explicit formulae in Theorem \ref{the:theorem} and \ref{thm-op}, we compute $\D \mathbb{P}(\max_{0 < t < 1} b_N(t) <M)$ numerically. We set the number of Bessel paths $N$ to be $10$ and the order of Bessel functions $\alpha$ to be 1 as an illustration.

As expected, when the starting point $a$ is fixed, the probability $\D \mathbb{P}(\max_{0 < t < 1} b_N(t) <M)$ is positive and increase with respect to $M$. When $M$ is large, the probability tends to 1; see Figure \ref{fig1} and Table \ref{table:M}. When the upper constraint $M$ is fixed, the probability decreases with respect to $a$, and tends to 0 when $a$ is large; see Figure \ref{fig2} and Table \ref{table:a}.

\begin{figure} [h]
  \begin{minipage}[t]{0.45\linewidth}
    \centering
    \includegraphics[width=3in]{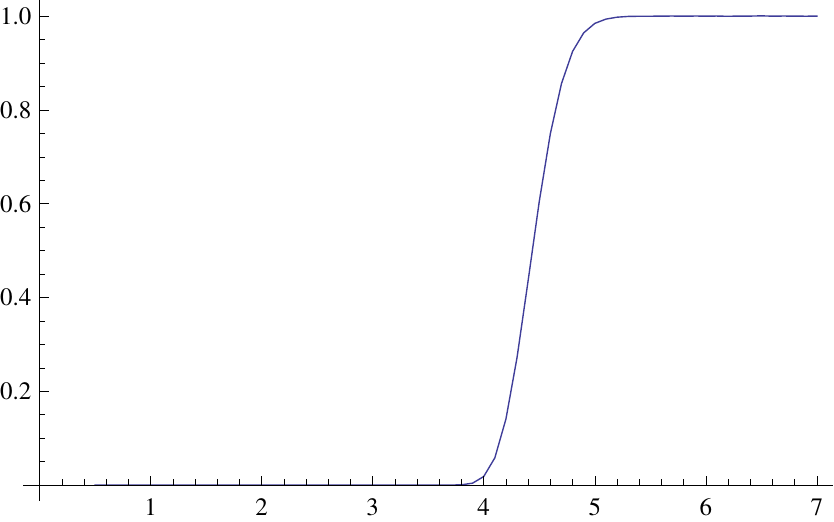}
     \caption{The evolution of the probability $\D \mathbb{P}(\max_{0 < t < 1} b_N(t) <M)$ with respect to $M$ when $N=10, \alpha = 1$ and $a = 1$.}
    \label{fig1}
  \end{minipage}
  \qquad
  \begin{minipage}[t]{0.45\linewidth}
    \centering
    \includegraphics[width=3in]{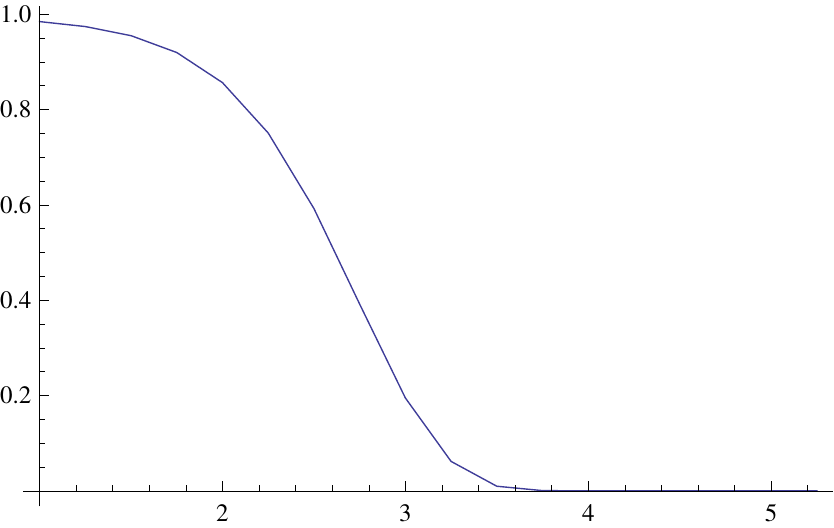}
    \caption{The evolution of the probability $\D \mathbb{P}(\max_{0 < t < 1} b_N(t) <M)$ with respect to $a$ when $N=10, \alpha = 1$ and $M=5$.}
    \label{fig2}
  \end{minipage}
\end{figure}

\begin{minipage}[b]{.45\textwidth}
  \centering
  \begin{tabular}{|c|c|}
\hline
M    & $\displaystyle \mathbb{P}(\max_{0<t<1}b_N(t)<M)$  \\ \hline
3.25 & $2.3576\times 10^{-10}$ \\
3.5  & $9.82109\times 10^{-8}$ \\
3.75 & 0.000233222                    \\
4    & 0.0183778                      \\
4.25 & 0.200657                       \\
4.5  & 0.604697                       \\
4.75 & 0.894626                       \\
5    & 0.984423                       \\
5.25 & 0.998446                       \\
5.5  & 0.999666                       \\
5.75 & 0.999961                       \\
6    & 0.999919                       \\
 \hline
\end{tabular}
\captionof{table}{The probability with $N=10, \alpha = 1$ and $a = 1$.}
   \label{table:M}
\end{minipage}\qquad
\begin{minipage}[b]{.45\textwidth}
   \centering
   \begin{tabular}{|c|c|}
\hline
$a$    & $\displaystyle \mathbb{P}(\max_{0<t<1}b_N(t)<M)$       \\ \hline
1    & 0.984373                        \\
1.25 & 0.974007                        \\
1.5  & 0.954902                        \\
1.75 & 0.919556                        \\
2    & 0.856389                        \\
2.25 & 0.750728                        \\
2.5  & 0.592312                        \\
2.75 & 0.391808                        \\
3    & 0.194887                        \\
3.25 & 0.0617876                       \\
3.5  & 0.00968417                      \\
3.75 & 0.000500336                     \\
  \hline
\end{tabular}
\captionof{table}{The probability with $N=10, \alpha = 1$ and $M = 5$.}
\label{table:a}

\end{minipage}

\subsection{Conclusions}

In summary, we derive the probability $\D \mathbb{P}(\max_{0 < t < 1} b_N(t) <M)$ for the maximum height of $N$ non-intersecting Bessel paths exactly. The probability can be put in terms of Hankel determinants associated with the multiple discrete orthogonal polynomials or discrete orthogonal polynomials, respectively, depending on the starting point $a>0$ or $a=0$. It is well-known that the Deift-Zhou nonlinear steepest descent method for Riemann-Hilbert problems is a powerful tool to study asymptotic problems related to (multiple) orthogonal polynomials. Therefore, this relation is important for us to derive the asymptotics of $\D \mathbb{P}(\max_{0 < t < 1} b_N(t) <M)$ as the number of paths $N \to \infty$. With suitable scaling, we expect the limiting distribution converges to the Tracy-Widom distribution for the largest eigenvalue of the Gaussian orthogonal ensemble due to its close relation to the Airy$_2$ process. We will leave the asymptotic study in a forthcoming publication.

\section*{Acknowledgments}

This work was partially supported by grants from the Research Grants Council of the Hong Kong Special Administrative Region, China (Project No. CityU 11300115, CityU 11303016), and by grants from City University of Hong Kong (Project No. 7004864, 7005032).

\end{document}